\documentclass{article}
\pdfoutput=1
\usepackage[a4paper]{geometry}
\usepackage{amsmath,amssymb,amsfonts,amsthm} 
\usepackage[font=scriptsize]{caption}
\usepackage{graphicx}
\usepackage{psfrag}
\usepackage{mathrsfs}
\usepackage{hyperref}
\usepackage{url}
\usepackage{bm}%
\usepackage{cite}

\def\A{{\cal A}}
\def\K{{\cal K}}

\def\R{\mathbb{R}}
\def\C{\mathbb{C}}
\def\Z{\mathbb{Z}}

\def\unitsphere{\mathbb{S}^2}

\DeclareMathOperator{\Id}{Id}

\newtheorem{definition}{Definition}
\newtheorem{lemma}{Lemma}

\begin{document}

\title{\bf On the stability of periodic $N$-body motions with the
  symmetry of Platonic polyhedra}

\author{M. Fenucci\footnote{Marco Fenucci, Dipartimento di Matematica,
    Universit\`a di Pisa, Largo B. Pontecorvo, 5, Pisa, Italy, {\tt
      fenucci@mail.dm.unipi.it}} and
  G. F. Gronchi\footnote{Giovanni~F. Gronchi, Dipartimento di
    Matematica, Universit\`a di Pisa, Largo B. Pontecorvo, 5, Pisa,
    Italy, {\tt gronchi@dm.unipi.it}}}


\maketitle

\begin{abstract}
In \cite{fgn10} several periodic orbits of the Newtonian $N$-body
problem have been found as minimizers of the Lagrangian action in
suitable sets of $T$-periodic loops, for a given $T>0$.  Each of them
share the symmetry of one Platonic polyhedron.
In this paper we first present an algorithm to enumerate all the
orbits that can be found following the proof in \cite{fgn10}. Then we
describe a procedure aimed to compute them and study their stability. Our
computations suggest that all these periodic orbits are unstable.  For
some cases we produce a computer-assisted proof of their instability
using multiple precision interval arithmetic.
\end{abstract}

\section{Introduction}

The existence of several periodic orbits of the Newtonian $N$-body
problem has been proved by means of variational methods, see
e.g. \cite{CM2000, hiphop, FT2004, terracini, terrvent07}.  In most
cases these orbits are found as minimizers of the Lagrangian action
functional and the bodies have all the same mass.  One difficulty with
the variational approach is that the Lagrangian action functional $\A$
is not coercive on the whole Sobolev space of $T$-periodic loops, for
which a natural choice is $H^1_T(\R,\R^3)$. We can overcome this
problem by restricting the domain of the action to symmetric loops or
by adding topological constraints, e.g. \cite{dgg89, bcz91, gordon77}.
Another difficulty is to prove that the minimizers are free of
collisions. For this purpose we can use different techniques, like
level estimates or local perturbations, see \cite{marchal01, ch02,
  chen03}.
We observe that the existence of periodic orbits with different
masses, minimizing the Lagrangian action, has been proved for the case
of three bodies \cite{Chen2008}. Moreover, also periodic orbits that are not
minimizers have been found using the mountain pass theorem
\cite{aribarterr}.

%
%

Besides the theoretical approach, also numerical methods have been used
to search for periodic motions in a variational context.  The first
evidence of the existence of a periodic orbit of the $3$-body problem
where three equal masses follow the same eight-shaped trajectory (the {\em Figure
  Eight}) can be found in \cite{moore93}.
Several periodic motions with a rich symmetry structure can be found
in \cite{simo:newfamilies, simo2000}, where the term {\em
  choreography} was first used to denote a motion of $N$ equal masses
on the same closed path equally shifted in phase. The
introduction of rigorous numerical techniques, see \cite{zgli:lohner},
led to computer-assisted proofs of the existence of periodic orbits in
dynamical systems \cite{kapela02}.
Numerical methods have been also used to study bifurcations and
stability of such periodic orbits. For example the linear and KAM stability
of the Figure Eight were first noticed in \cite{simo2002}. Later on these
stability results were made rigorous with a computer-assisted proof
\cite{kapela:eight, kapela:KAM}.

In this paper we focus on periodic motions of the Newtonian $N$-body
problem, with equal masses, sharing the symmetry of Platonic
polyhedra. In particular we present an algorithm to enumerate all the
orbits that can be found following the proof in \cite{fgn10}, which
minimize the Lagrangian action in suitable sets of $T$-periodic loops,
for a given $T>0$.  Then we describe a procedure aimed to compute them and
study their stability. Our computations suggest that all these
periodic orbits are unstable.  For some cases we produce a
computer-assisted proof of their instability using multiple precision
interval arithmetic.

The paper is organized as follows. In Section~\ref{s:exist} we recall
the steps of the existence proof of non-collision minimizers of the action $\A$.
In Sections~\ref{s:enum}, \ref{s:numcomp} we present a method to
enumerate all the periodic orbits and to compute them. The linear
stability theory for this case is briefly reviewed in
Section~\ref{s:ls}. In Section~\ref{s:valid} we describe a procedure
to check the conditions for stability with rigorous numerics, and we
perform a computer-assisted proof of the instability for some of these
orbits.

\section{Proving the existence of non-collision minimizers}
\label{s:exist}

We recall the steps of the proof of the existence of the periodic
orbits given in \cite{fgn10}. Let us fix a positive number $T$ and let
${\cal R}$ be the rotation group of one of the five Platonic
polyhedra.  We consider the motion of $N=|{\cal R}|$ particles
with unitary mass.
Let us denote by $u_I:\R\to\R^3$ the map describing the
motion of one of these particles, that we call {\em generating particle}.
Assume that
\begin{itemize}
\item[(a)] the motion $u_R$, $R\in{\cal R}\setminus\{I\}$ of the other
  particles fulfills the relation
\begin{equation}
  u_R = Ru_I,
  \label{symcond}
\end{equation}

\item[(b)] the trajectory of the generating particle belongs to a
  given non-trivial free homotopy class of $\R^3\setminus\Gamma$,
  where
\[
  \Gamma=\displaystyle\cup_{R\in{{\mathcal R}\setminus\{I\}}}r(R),
\]
with $r(R)$ the rotation axis of $R$.

\item[(c)] there exist $R\in{\cal R}$ and $M>0$ such that 
\begin{equation}
u_I(t+T/M) = Ru_I(t),
\label{extrasym}
\end{equation}
for all $t \in \R$.

\end{itemize}
Imposing the symmetry \eqref{symcond}, the action functional of the
$N$-body problem depends only on the motion of the generating particle
and it is expressed by
\begin{equation}
\A(u_I) = N\int_0^T \biggl(\frac{1}{2}|\dot{u}_I|^2 + \frac{1}{2}\sum_{R\in{\cal
R}\setminus\{I\}}\frac{1}{|(R-I)u_I|}\biggr)\,dt.
\label{action_formula}
\end{equation}
We search for periodic motions by minimizing $\A$ on subsets $\K$ of
a Sobolev space of $T$-periodic maps. More precisely we choose the cones
\begin{equation}
\K = \{ u_I \in H_T^1(\R,\R^3): (b), (c) \mbox{ hold} \}.
\label{conesK}
\end{equation}

\subsection{Encoding $\K$ and existence of minimizers}

We describe two ways to encode the topological constraints defining
the cones $\K$. Let $\tilde{\cal R}$ be the full symmetry group
(including reflections) related to ${\cal R}$.
The reflection planes induce a tessellation of the unit sphere
$\unitsphere$, as shown in Figure~\ref{fig:tessel}, with $2N$ spherical
triangles.
\begin{figure}[!h]
   \centerline{\includegraphics[width=4.5cm]{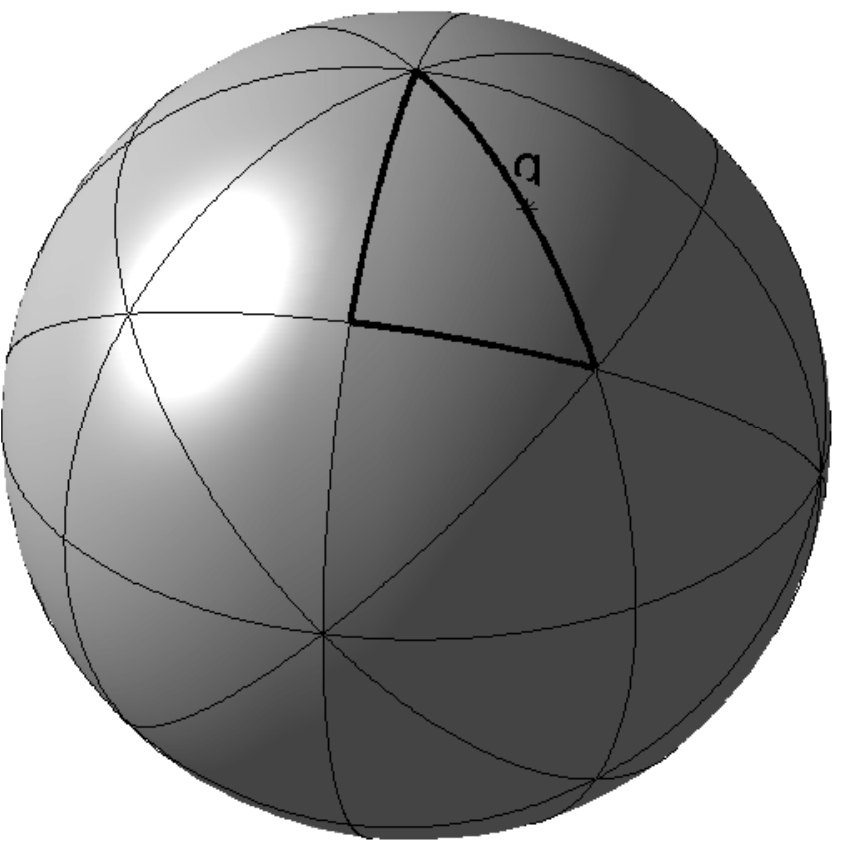}
\hskip 1.5cm
\includegraphics[width=4.1cm]{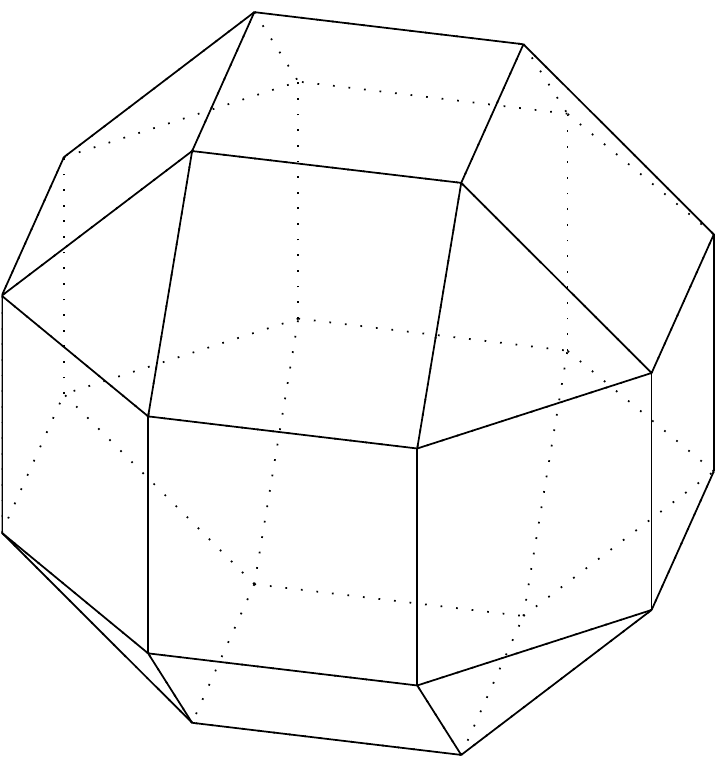}}
\caption{Tessellation of $\unitsphere$ for ${\cal R}={\cal
    O}$ and the Archimedean polyhedron ${\cal Q}_{\cal O}$.}
\label{fig:tessel}
\end{figure}
Each vertex of such triangles corresponds to a pole $p\in{\cal P} =
\Gamma\cap\unitsphere$. Let us select one triangle, say $\tau$. By a
suitable choice of a point $q\in\partial\tau$ (see
Figure~\ref{fig:tessel}) we can define an Archimedean polyhedron
${\cal Q}_{\cal R}$, which is the convex hull of the orbit of $q$
under ${\cal R}$, and therefore it is strictly related to the symmetry
group ${\cal R}$. For details see \cite{fgn10}.

\noindent We can characterize a cone ${\cal K}$ by a periodic sequence
$\mathfrak{t}=\{\tau_k\}_{k\in\Z}$ of triangles of the tessellation
such that $\tau_{k+1}$ shares an edge with $\tau_k$ and
$\tau_{k+1}\neq \tau_{k-1}$ for each $k\in\Z$. This sequence is
uniquely determined by $\K$ up to translations, and describes the
homotopy class of the admissible paths followed by the generating
particle (see Figure~\ref{fig:encoding}, left).


\noindent We can also characterize $\K$ by a periodic sequence
$\nu=\{\nu_k\}_{k\in\Z}$ of vertexes of ${\cal Q}_{\cal R}$ such that
the segment $[\nu_k,\nu_{k+1}]$ is an edge of ${\cal Q}_{\cal R}$ and
$\nu_{k+1}\neq \nu_{k-1}$ for each $k\in\Z$. Also the sequence $\nu$
is uniquely determined by $\K$ up to translations, and with it we can
construct a piecewise linear loop $v$, joining consecutive vertexes
$\nu_k$ with constant speed, that represents a possible motion of the
generating particle (see Figure~\ref{fig:encoding}, right).

\begin{figure}[!h]
\centerline{
   \includegraphics[width=4.5cm]{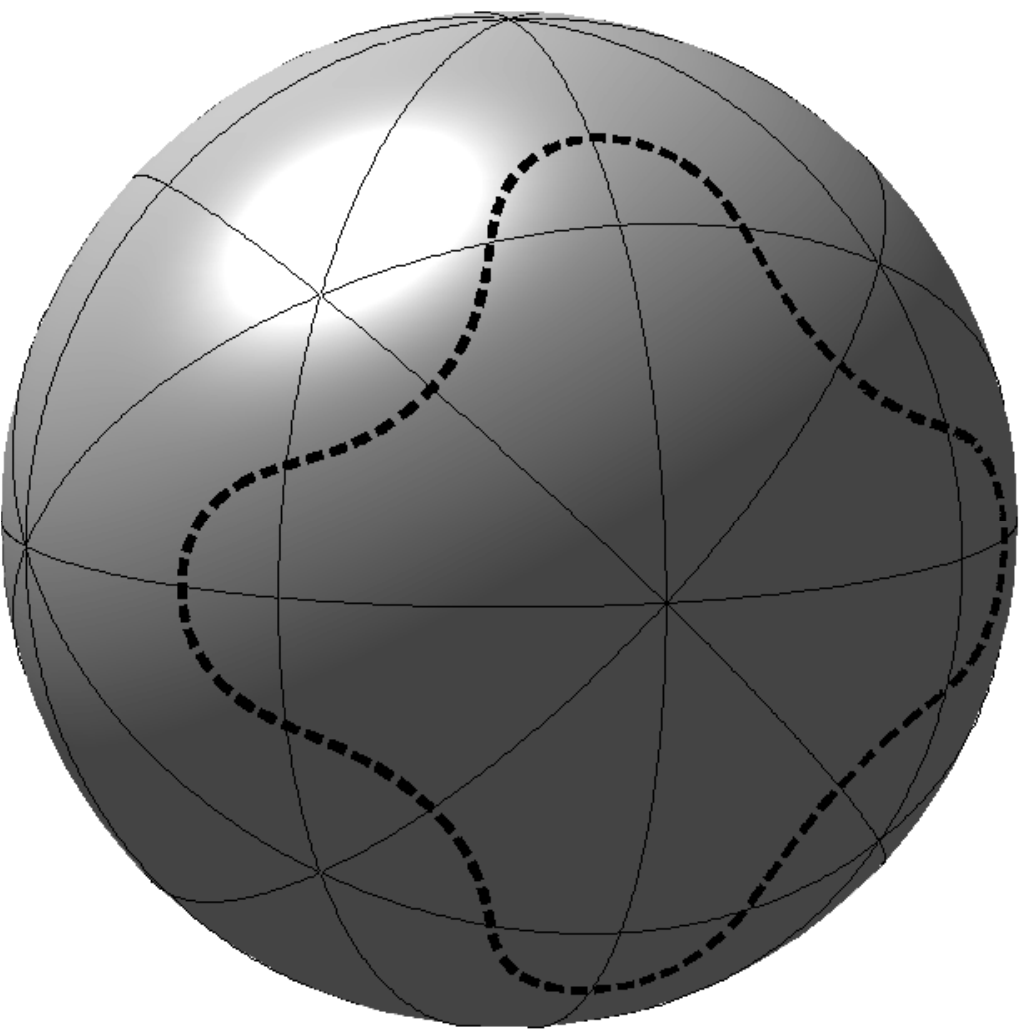}
\hskip 1.5cm
\includegraphics[width=4.2cm]{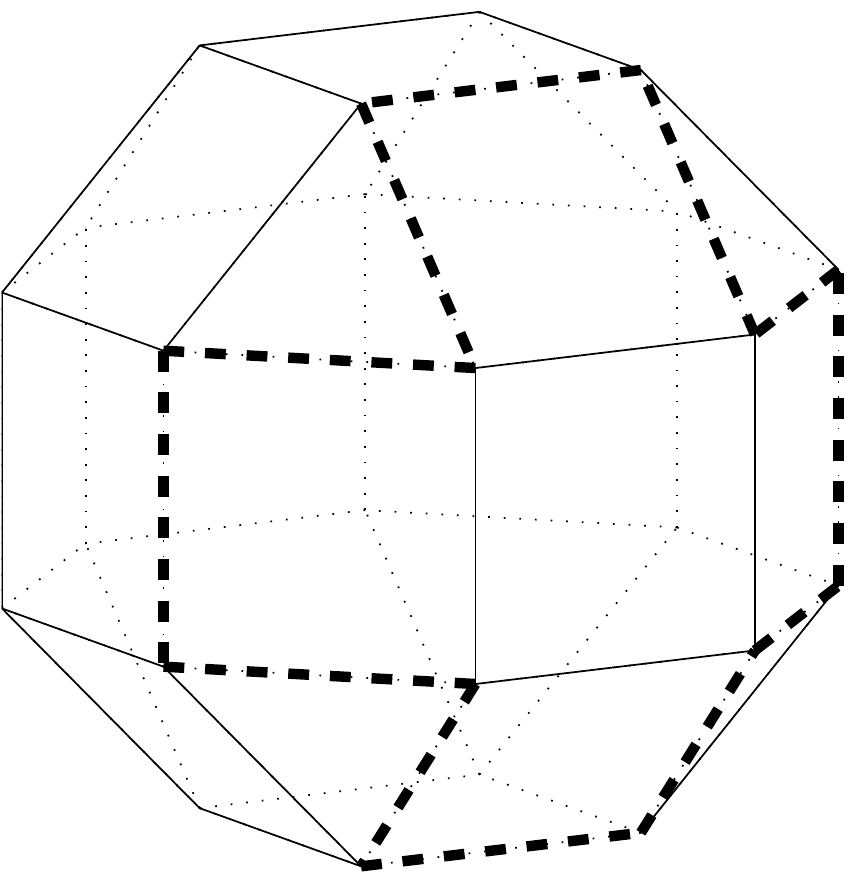}}
\caption{Encoding a cone $\K$. Left: the dashed path on $\unitsphere$
  describes the periodic sequence $\mathfrak{t}$ of triangles of the
  tessellation. Right: the dashed piecewise linear path describes the
  corresponding periodic sequence $\nu$ of vertexes of ${\cal Q}_{\cal
    O}$.}
\label{fig:encoding}
\end{figure}

The existence of a minimizer $u^*_I$ of $\A$ restricted to a cone $\K=\K(\nu)$
can be shown by standard methods of calculus of variations, provided
that
\begin{equation}
\bigcap_{\tau_j\in\mathfrak{t}}\overline{\tau_{j}} = \emptyset,
\label{nonunosolo}
\end{equation}
where $\mathfrak{t}$ is the sequence of spherical triangles
corresponding to $\K$. Condition (\ref{nonunosolo}) means that the
trajectory of the generating particles does not wind around one
rotation axis only: this ensures the coercivity of the action
functional and therefore a minimizer exists.

For later use we introduce the following definitions.
\begin{definition}
We say that a cone $\K$ is `simple' if the corresponding sequence
$\mathfrak{t}$ does not contain a string $\tau_k\ldots \tau_{k+2\mathfrak{o}}$
such that
\[
\bigcap_{j=0}^{2\mathfrak{o}}\overline{\tau_{k+j}} = p,
\]
where $p\in\cal P$ and $\mathfrak{o}$ is the order of $p$.
\label{def:simpCone}
\end{definition}

\begin{definition}
  We say that a cone $\K$ winds around two coboundary axes
  if
  \begin{itemize}
  \item[i)] the corresponding sequence $\mathfrak{t}$ is the union of
    two strings, $\tau_{k_j}\ldots \tau_{k_j+2\mathfrak{o}_j}$,
    $j=1,2$, such that
    \[
\bigcap_{h=0}^{2\mathfrak{o}_j}\overline{\tau_{k_j+h}} = p_j, 
\]
where $\mathfrak{o}_j$ is the order of $p_j$, for two different poles
$p_1, p_2$;
  \item[ii)] there exists $\tau_k\in\mathfrak{t}$ such that $p_1, p_2\in
    \overline{\tau_k}$.
  \end{itemize}
  \label{def:2ax}
\end{definition}

To show that for a suitable choice of $\K$ the minimizers are
collision-free we consider total and partial collisions separately.

\subsection{Total collisions}
We note that a total collision of the $N$ particles occurs at time
$t_c$ iff $u_I(t_c) = 0$.  If there is a total collision then, by
condition (c), there are $M$ of them per period. For a minimizer
$u^*_I$ with a total collision we can give the following \textit{a priori}
estimate  for the action (see \cite[Section 5]{fgn10}):
\begin{equation}
   \A(u^*_I) \geq \alpha_{\cal R, M},
   \label{eq:aPrioriEst}
\end{equation}
where $\alpha_{\cal R, M}$ depends only on $M$ and $T$. Rounded values
of $\alpha_{\cal R, M}$ for $T=1$ are given in Table \ref{tab:aR}.
\begin{table}[!h]
\begin{center}
\begin{tabular}{l|c|c|c|c|c}
{{\small ${\cal R}$}\raisebox{0.1cm}{$\diagdown$}\raisebox{0.2cm}{\small$M$}}
&\raisebox{0.2cm}{1}  &\raisebox{0.2cm}{2}  &\raisebox{0.2cm}{3}  
&\raisebox{0.2cm}{4}  &\raisebox{0.2cm}{5}\\
\hline
${\cal T}$  &$132.695$   &$210.640$   &$276.017$ &/  &/\\
\hline
${\cal O}$  &$457.184$   &$725.734$   &$950.981$  &$1152.032$ &/\\
\hline
${\cal I}$ &$2296.892$  &$3646.089$  &$4777.728$  &/  &$6716.154$\\
\end{tabular}
\caption{Lower bounds $a_{{\cal R},M}$ for loops with $M$ total
  collisions ($T=1$).}
\label{tab:aR}
\end{center}
\end{table}
For some sequences $\nu$, the action of the related piecewise linear
loop $v$ is lower than $\alpha_{\cal R, M}$. Therefore, minimizing the
action over the cones $\cal K$ defined by such sequences yields
minimizers without total collisions.  The action of the piecewise
linear loop $v$ can be computed explicitly and it is given by
\begin{equation}
   \A(v) = \frac{3}{2\cdot 4^{1/3}} N \ell^{2/3}(k_1\zeta_1+k_2\zeta_2)^{2/3},
   \label{eq:actLPL}
\end{equation}
where $k_1,k_2$ are the numbers of sides of the two different kinds
(i.e. separating different pairs of polygons) in the trajectory of
$v$, $\ell$ is the length of the sides (assuming ${\cal Q}_{\cal R}$
is inscribed in the unit sphere) and $\zeta_1,\zeta_2$ are the values
of explicitly computable integrals, see Table
\ref{tab:upsilon}. Relations \eqref{eq:aPrioriEst} and
\eqref{eq:actLPL} will be useful later in Section 2

\begin{table}[!h]
\begin{center}
\begin{tabular}{l|c|c|c}
${\cal R}$  &${\cal T}$ &${\cal O}$ &${\cal I}$ \\
\hline
$\ell$        &$1.0$ &$0.7149$ &$0.4479$ \\
\hline
$\zeta_1$  &$9.5084$  &$20.3225$ &$53.9904$ \\
\hline
$\zeta_2$  &$9.5084$  &$19.7400$ &$52.5762$ 
\end{tabular}
\caption{Numerical values of $\ell$, $\zeta_1$, $\zeta_2$.}
\label{tab:upsilon}
\end{center}
\end{table}

\subsection{Partial collisions}
Because of the symmetry a partial collision occurs at time $t_c$ iff
$u_I(t_c)\in\Gamma\setminus\{0\}$, that is when the generating
particle passes through a rotation axis $r$. Indeed, in this case all
the particles collide in separate clusters, each containing as many
particles as the order of $r$.  We summarize below the technique used
in \cite{fgn10} to deal with partial collisions.  We can associate to
a partial collision two unit vectors $\mathsf{n}^+, \mathsf{n}^-$,
orthogonal to the collision axis $r$, corresponding to the ejection
and collision limit directions respectively. By means of these vectors
we can define a \textit{collision angle} $\theta$ and, assuming that
$\cal K$ is {\em simple}, we have
\[
-\frac{\pi}{\mathfrak{o}_r} \leq \theta \leq 2\pi,
\]
where $\mathfrak{o}_r$ is the order of the maximal cyclic group
related to the collision axis $r$.
If $\theta
\neq 2\pi$ we can exclude partial collisions by local perturbations,
constructed by using either direct or indirect arcs \cite{ch05}, and
with a blow up technique \cite{FT2004}. If $\theta = 2\pi$, then
\begin{itemize}
   \item[i)] $\mathsf{n}^+ = \mathsf{n}^-$,
   \item[ii)] the plane $\pi_{r,\mathsf{n}}$ generated by $r$ and $\mathsf{n}=\mathsf{n}^{\pm}$
     is fixed by some reflection $\tilde{R} \in \tilde{\cal R}$,
\end{itemize}
and we say that the partial collision is of type
$(\rightrightarrows)$. In this case we cannot exclude the singularity
by a local perturbation because the indirect arc is not
available. However, it turns out that in this case the trajectory of
the generating particle must lie on a reflection plane, bouncing
between two coboundary rotation axes.

\noindent We conclude that, provided that $\cal K$ is simple and it
does not wind around two coboundary axes, the minimizer $u^*_I$ of the
action $\A$ restricted to $\K$ is free of partial collision, hence it
is a smooth periodic solution of the $N$-body problem.

\section{Enumerating all the collision-free minimizers}
\label{s:enum}

Here we introduce an algorithm to generate all the sequences $\nu$ of
length $l$, for some admissible integer $l$. Then we select only the
periodic ones, and control whether they satisfy all the conditions
ensuring the existence of collision-free minimizers of
(\ref{action_formula}) in the corresponding cone $\K=\K(\nu)$. Precisely,
our algorithm is based on the following steps:
\begin{enumerate}
  \item Find the maximal admissible length $l_{\text{max}}$ and other
      constraints on the length $l$.

  \item Construct all the periodic sequences $\nu$ of vertexes of the
     Archimedean polyhedron $\mathcal{Q}_{\mathcal{R}}$.

  \item Exclude the sequences that wind around one axis only or around
    two coboundary axes.

  \item Exclude the sequences that do not respect the additional
     choreography symmetry \eqref{extrasym}.
  \item Exclude the sequences that give rise to non simple cones.

\end{enumerate}

\subsection{Constraints on the length}

To exclude total collisions we use \eqref{eq:aPrioriEst} and
\eqref{eq:actLPL}.  If $v$ is the piecewise linear loop defined by the
sequence $\nu$, the relation $\A(v) < \alpha_{\mathcal{R},M}$ can be
rewritten as
\begin{equation}
   k_1\zeta_1 + k_2\zeta_2 < \bigg( \alpha_{\mathcal{R}, M} \frac{2\cdot
   4^{1/3}}{3T^{1/3}}\frac{1}{N}\frac{1}{\ell^{2/3}} \bigg)^{3/2} := K.
   \label{eq:actionPiecewise}
\end{equation}
Since the coefficients $\zeta_1$ and $\zeta_2$ are positive, for each
$M$ there exist only a finite number of positive integers $(k_1,k_2)$
fulfilling (\ref{eq:actionPiecewise}).
Given ${\cal R}\in\{{\cal T}, {\cal O}, {\cal I}\}$, taking the
maximal value of $k_1+k_2$ we get a constraint on the maximal length
$l_{\text{max}} = l_{\text{max}}(M)$ of the sequence $\nu$, see
Table~\ref{tab:lmax}.
\begin{table}[!h]
\begin{center}
\begin{tabular}{l|c|c|c|c|c}
 {\small ${\cal R}$}\raisebox{0.1cm}{$\diagdown$}\raisebox{0.2cm}{\small $M$}
&\raisebox{0.2cm}{1}  &\raisebox{0.2cm}{2}  &\raisebox{0.2cm}{3}  
&\raisebox{0.2cm}{4}  &\raisebox{0.2cm}{5}\\
\hline
${\cal T}$  &$4$   &$8$   &$12$ &/  &/\\
\hline
${\cal O}$  &$6$   &$12$   &$19$  &$25$ &/\\
\hline
${\cal I}$ &$10$  &$21$  &$32$  &/  &$54$
\end{tabular}
\caption{Values of $l_{\text{max}}(M)$ for the different symmetry groups.}
\label{tab:lmax}
\end{center}
\end{table}

On the other hand, we can also give a constraint on the minimal length
$l_{\rm min}$. In fact, a periodic sequence of length $l\leq 5$ either
winds around one axis only or encloses two coboundary axes.  However,
for all the five Platonic polyhedra there exists at least a good
sequence $\nu$ of length $6$: for this reason we set $l_{\rm min} =
6$.

Furthermore, in the case of $\mathcal{R}=\mathcal{T},\mathcal{O}$ we
cannot have $M = 1$. In fact:
\begin{itemize}
\item[-] if $\mathcal{R}=\mathcal{T}$,  $l_{\text{max}}(1)=4$,
  therefore we cannot construct any good sequence $\nu$.
\item[-] if $\mathcal{R}=\mathcal{O}$, $l_{\text{max}}(1)=6$, and the only sequences of length $6$ that do not wind
      around two coboundary axes have $M=2$.
\end{itemize}

\subsection{Periodic sequences construction}
To know which vertexes are reachable from a fixed vertex $V_j$ of
$\mathcal{Q}_{\mathcal{R}}$, we interpret the polyhedron as a connected graph: in
this manner we have an \textit{adjacency matrix} $A$ associated to the graph. In this
matrix we want to store the information about the kind of sides connecting two different
vertexes.
The generic entry of $A$ is
\begin{equation}
   A_{ij} = 
   \begin{cases}
      1 & \text{if the vertex $i$ and the vertex $j$ are connected by a side of type 1}, \\
      2 & \text{if the vertex $i$ and the vertex $j$ are connected by a side of type 2}, \\
      0 & \text{otherwise.}
   \end{cases}
   \label{eq:adj_matrix}
\end{equation}
For a fixed length $l \in \{ l_{\rm min},\dots,l_{\rm max} \}$,
we want to generate all the sequences of vertexes with that length,
starting from vertex $1$. Because of the symmetry, we can select the first side
arbitrarily, while in the other steps we can
choose only among $3$ different vertexes, since we do not want to
travel forward and backward along the same side.  Therefore, the
total number of sequences with length $l$ is $3^{l-1}$. To generate
all these different sequences we produce an array of choices
$c=(c_1,\dots,c_l)$ such that $c_1 = 1$ and
$c_j \in \{ 1,2,3 \}, j=2,\dots,l$. Each entry tells us the
way to construct the sequence: if $v_1,v_2,v_3$ are the number of the
vertexes reachable from $\nu_j$ (with $v_i$ sorted in ascending
order), then $\nu_{j+1} = v_{c_j}$. All the different $3^{l-1}$
arrays of choices can be generated using an integer number $k \in
\{0,\dots,3^{l-1}-1 \}$, through its base $3$ representation.

\subsection{Winding around one axis only or two coboundary axes}
To check whether a closed sequence winds around one axis only we have
to take into account the type of Archimedean polyhedron
$\mathcal{Q}_{\mathcal{R}}$.  In the cases
$\mathcal{R}=\mathcal{T},\mathcal{O}$ it is sufficient to count the
number $m$ of different vertexes appearing in $\nu$. If $m=3,4$ then
$\nu$ winds around one axis only. The case of
$\mathcal{R}=\mathcal{I}$ is different, since also pentagonal faces
appear. If $m=5$, to check this property we can take the mean of the
coordinates of the touched vertexes and control whether it coincides
with a rotation axis or not.

We note that for $M$ different from $1$, a periodic sequence
satisfying the choreography condition \eqref{eq:choreoCond}, introduced
in the next paragraph, cannot wind around two coboundary axes. For this
reason we decided to avoid performing this additional control, since
$M=1$ is possible only in the case of $\mathcal{R}=\mathcal{I}$. In
this case we exclude the non-admissible sequences in a non-automated
way, looking at them one by one.  We point out that such a sequence
can be of three different types:
\begin{enumerate}
   \item a pentagonal face and a triangular face sharing a vertex;
   \item a pentagonal face and a square face sharing a side;
   \item a square face and a triangular face sharing a side.
\end{enumerate}
The sequences that travel along the boundary of two square faces sharing a 
vertex winds around two axes too, but these axes are not coboundary, 
thus we have to keep them.

\subsection{Choreography condition} 
Condition \eqref{extrasym} is satisfied if and only if there
exists a rotation $R \in \mathcal{R}$ such that
\begin{equation}
   \eta_R\nu_j = \nu_{j+k},
   \label{eq:choreoCond}
\end{equation}
for some integer $k$, where $\eta_R$ denotes the permutation of the
vertexes of $\mathcal{Q}_{\mathcal{R}}$ induced by $R$. To check 
condition \eqref{eq:choreoCond}, we have to construct 
$\eta_R$. Let $V_1,\dots,V_N \in \R^3$ be the coordinates
of all the vertexes of $\mathcal{Q}_{\mathcal{R}}$: since each rotation $R$ leaves
$\mathcal{Q}_{\mathcal{R}}$ unchanged, it sends vertexes into vertexes.
Therefore, we construct the matrices $\eta_R$ such that
\[
(\eta_R)_{ji} =
\begin{cases}
   1 & \text{if } RV_i = V_j, \\
   0 & \text{otherwise}.
\end{cases}
\]
It results that each $\eta_R \in \R^{N \times N}$ is a \textit{permutation
  matrix}.  The product of $\eta_R$ with the vector $\textbf{v} =
(1,\dots,N)^T$ provides the permutation of $\{ 1,\dots,N \}$.
At this point, given a rotation $R$, we are able to write the
permuted sequence $\eta_R\nu$: we can simply check that
\eqref{eq:choreoCond} holds by comparing the sequences $\nu$ and $\eta_R
\nu$. Moreover, if the condition is satisfied, we compute the value
$M=k_\nu/k$, where $k_\nu$ is the minimal period of $\nu$.


\subsection{Simple cone control}
Definition \ref{def:simpCone} of simple cones is given by using the
tessellation of the sphere induced by the reflection planes of the
Platonic polyhedra.  To decide whether a cone $\K$
is simple or not, we must translate this definition into a condition
on the sequence of vertexes. We observe that the only way to produce a
non-simple cone is by traveling all around the boundary of a face
$\mathcal{F}$ of $\mathcal{Q}_{\mathcal{R}}$:
for the cone $\K$, being simple or not depends on the order of the
pole associated to $\mathcal{F}$ and on the way the oriented path
defined by $\nu$ gets to the boundary of $\mathcal{F}$ and leaves it.
We discuss first the case of a triangular face, pointing out that it
is associated to a pole $p$ of order three. Suppose that $\nu$ contains a
subsequence $[\nu_k, \, \nu_{k+1}, \, \nu_{k+2}, \, \nu_{k+3}]$ that
travels all around a triangular face $\mathcal{F}$, that is $\nu_k =
\nu_{k+3}$. Let $\nu_{k-1}$ and $\nu_{k+4}$ be the vertexes before and
after accessing the boundary of $\mathcal{F}$. Four different
cases can occur:
\begin{itemize}
   \item[i)] $\nu_{k+4}$ is a vertex of the triangular face $\mathcal{F}$
     (Fig.~\ref{fig:tri}, top left);
   \item[ii)] the path defined by $\nu$ accesses and leaves $\mathcal{F}$
     through the same side, i.e. $[\nu_{k-1},\nu_k] =
     [\nu_{k+3},\nu_{k+4}]$ (Fig.~\ref{fig:tri}, top right);
   \item[iii)] the path defined by $\nu$ accesses and leaves $\mathcal{F}$
     through two different sides describing an angle $\theta >\pi$ around $p$
     (Fig.~\ref{fig:tri}, bottom left);
   \item[iv)] the path defined by $\nu$ accesses and leaves $\mathcal{F}$
     through two different sides describing an angle $\theta<\pi$ around $p$
     (Fig.~\ref{fig:tri}, bottom right).
\end{itemize}
\begin{figure}[!ht]
   \begin{center}
      \includegraphics[scale=0.4]{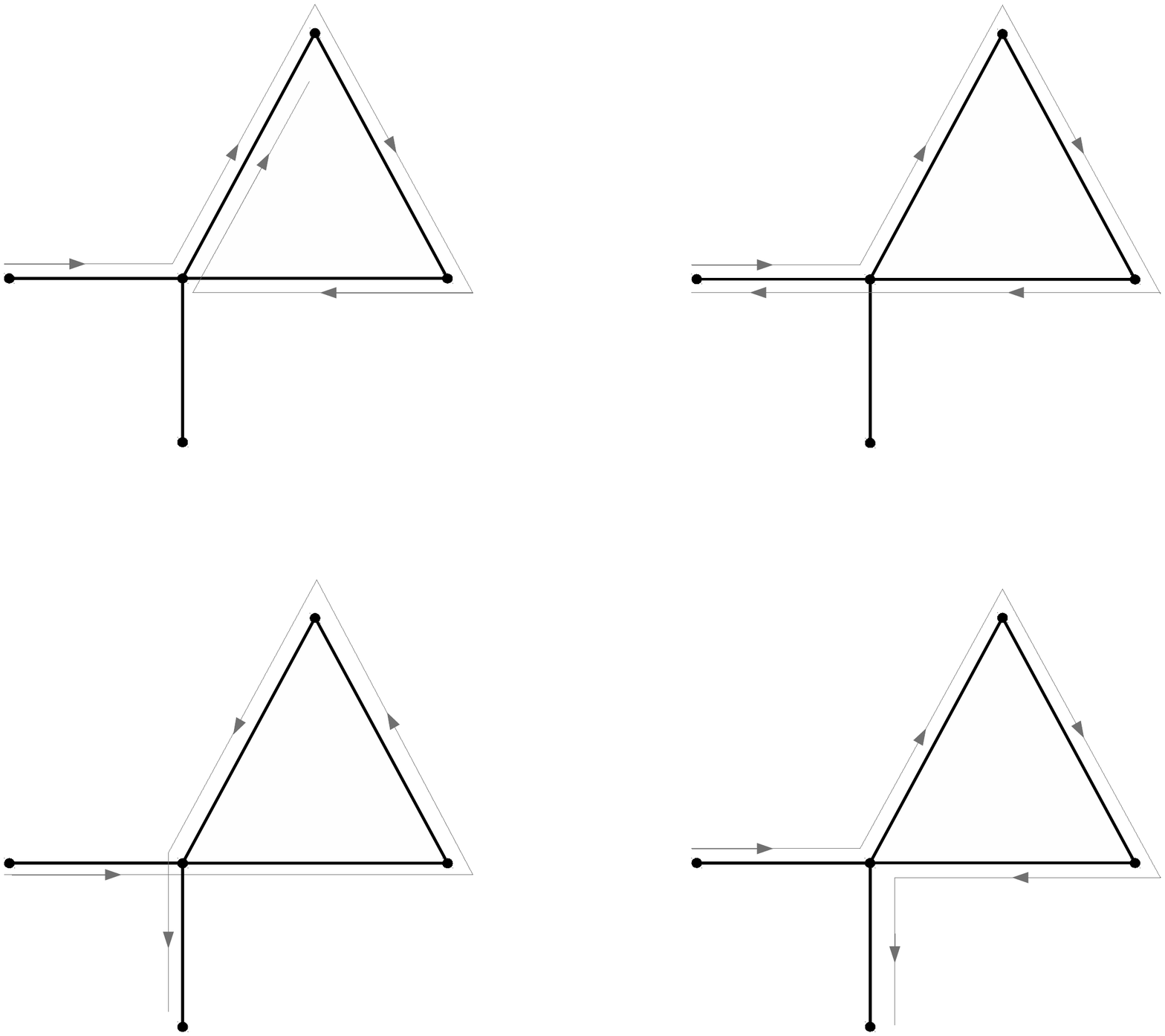}
   \end{center}
   \caption{The different cases occurring for a path around a triangular face of
      $\mathcal{Q}_{\mathcal{R}}$.}
   \label{fig:tri}
\end{figure}

In the first three cases, the sequence $\nu$ defines a non-simple
cone, while case iv) is the only admissible situation for a simple
cone.
Cases iii) and iv) can be distinguished by the sign of
$\vec{a}\times\vec{d}\cdot\vec{b}\times\vec{c}$
where
\[
\vec{a} = [\nu_{k-1},\nu_k], \quad
\vec{b} = [\nu_k,\nu_{k+1}], \quad
\vec{c} = [\nu_{k+2},\nu_{k+3}], \quad
\vec{d} = [\nu_{k+3},\nu_{k+4}].
\]
If this sign is positive we have case iii), if negative case iv).
This argument concludes the discussion about triangular faces.
Actually, we can see that the same argument can be used for square and
pentagonal faces, provided that the associated poles have order
greater than two. However, in the case $\mathcal{R} = \mathcal{O}$, we
can find poles of order two associated to square faces. In this
situation there is no way to travel all around the square face and get
a simple cone.

\subsection{Summary of the procedure} 
\label{s:algor}
Now we summarize the procedure that we adopt.  For each admissible
value of $M$ (see Table~\ref{tab:aR}) we observe that $M$ divides the
possible lengths $l \in \{ l_{\rm min}, \dots, l_{\text{max}}(M) \}$
of the sequence $\nu$.  Then, for each integer $h \in \{ 0,\dots,3^{l/M-1}-1
\}$ we perform the following steps:
\begin{enumerate}
   \item construct the array of choices $c$ corresponding to $h$;
   \item generate the sequence $\hat{\nu}$ on the basis of $c$,
     starting from vertex number $1$. Note that $\hat{\nu}$ is a sequence
     with length $l/M+1$;
   \item control whether $\eta_R\hat{\nu}_1 = \hat{\nu}_{l/M+1}$ for
     some $R \in \mathcal{R}$. In this case we extend $\hat{\nu}$ to a
     sequence $\nu$ with length $l+1$, using the choreography
     condition \eqref{eq:choreoCond};
   \item check whether $\nu$ is periodic or not;
   \item compute the minimal period $k_\nu$ of $\nu$;
   \item check whether $\nu$ winds around one axis only or not;
   \item check whether \eqref{eq:choreoCond} holds or not; if it
     holds, check whether the ratio $k_\nu/k$ is equal to $M$ or not;
   \item compute the values of $k_1, k_2$ and check whether
     \eqref{eq:actionPiecewise} holds or not;
   \item check whether the cone $\K=\K(\nu)$ is simple or not.
\end{enumerate}
If the sequence $\nu$ passes all the controls above, then there exists
a collision-free minimizer of the action $\A$ restricted to $\K$.

\subsection{Results}

The lists of good sequences found by the algorithm described in
Section~\ref{s:algor} for the three groups ${\cal T}$, ${\cal O}$,
${\cal I}$ are available at \cite{MFweb}. Here we
list only the total number of good sequences (i.e. leading to
collision-free minimizers and then classical periodic orbits) for the
different polyhedra, in Table \ref{tab:totNum}.
\begin{table}[!ht]
\begin{center}
\begin{tabular}{|c|c|}
\hline
$M$ & \text{Total number} \\
\hline
$2$ & $3$ \\
$3$ & $6$ \\
\hline
    & $9$ \\
\hline
\end{tabular}
\quad
\begin{tabular}{|c|c|}
\hline
$M$ & \text{Total number} \\
\hline
$2$ & $24$ \\
$3$ & $18$ \\
$4$ & $15$ \\
\hline
    & $57$ \\
\hline
\end{tabular}
\quad
\begin{tabular}{|c|c|}
\hline
$M$ & \text{Total number} \\
\hline
$1$ & $28$ \\
$2$ & $386$ \\
$3$ & $455$ \\
$5$ & $573$ \\
\hline
    & $1442$ \\
\hline
\end{tabular}
\caption{Total number of sequences $\nu$ found for $\mathcal{Q}_{\mathcal{T}},
\mathcal{Q}_{\mathcal{O}}, \mathcal{Q}_{\mathcal{I}}$ respectively, from the left to the
right.} 
\label{tab:totNum}
\end{center}
\end{table}
All the periodic orbits listed in \cite{fgn10} were found again with
this procedure.

We point out that we can identify two different sequences
$\nu,\tilde{\nu}$ if there exists a symmetry $S$ of the polyhedron
$\mathcal{Q}_{\mathcal{R}}$ such that $\eta_R \nu = \tilde{\nu}$,
where $\eta_R$ still denotes the permutation of the vertexes induces
by the symmetry $S$. The results are presented using this
identification.

\section{Numerical computation of the orbits}
\label{s:numcomp}

We describe the procedure that we have used to compute the periodic
orbits described in the previous sections.  Given the sequence $\nu$,
first we search for a Fourier polynomial approximating the minimizer
of the action functional $\A$ in $\K=\K(\nu)$. Then we refine the
approximation by a multiple shooting method that takes into account
the symmetry of the orbit.

\subsection{Approximation with Fourier polynomials}
Following \cite{mn08} and \cite{simo:newfamilies}, we want to find an
approximation of the motion minimizing the action functional $\A$.
To discretize the infinite dimensional set of $T$-periodic loops, we take into account the
truncated Fourier series at some order $F_M$.
We consider only loops $u: [0,T] \to \R^3$ of the form
\begin{equation}
   u(t) = \frac{a_0}{2} + \sum_{k=1}^{F_M} \bigg{[} a_k \cos\bigg{(} \frac{2\pi k }{T} t
   \bigg{)} + b_k \sin \bigg{(} \frac{2\pi k}{T} t \bigg{)}   \bigg{]},
   \label{eq:fourierSum}
\end{equation}
where $a_k, b_k \in \R^3$ are the Fourier coefficients.  
Restricting $\A$ to loops of the form \eqref{eq:fourierSum} we obtain a
function $A: \R^{3(2F_M+1)} \to \R$ that discretizes the action
functional. This function is defined (neglecting the constant factor $N$)
by
\begin{equation}
   A(a_0,a_1,\dots,a_{F_M}, b_1,\dots,b_{F_M}) = \int_{0}^{T}
   \bigg( \frac{1}{2} \lvert\dot{u} \rvert^2 + U(u) \bigg) dt,
   \label{eq:actionDisc}
\end{equation}
where
\[
   U(u) = \frac{1}{2}\sum_{R\in{\cal R}\setminus\{I\}}\frac{1}{|(R-I)u|}.
\]
The derivatives with respect to the Fourier coefficients are
\begin{gather}
   \frac{\partial A}{\partial a_k } = \frac{2(\pi k)^2}{T} a_k + \int_{0}^{T}
   \frac{\partial U}{\partial u}(u(t)) \cos\bigg{(} \frac{2\pi k}{T}t \bigg{)}dt, \quad k \geq 0, 
   \label{eq:partialA}\\ 
   \frac{\partial A}{\partial b_k } = \frac{2(\pi k)^2}{T} b_k + \int_{0}^{T}
   \frac{\partial U}{\partial u}(u(t)) \sin\bigg{(} \frac{2\pi k}{T}t \bigg{)}dt, \quad k
   > 0.
   \label{eq:partialB}
\end{gather}
Note that these derivatives may be large for high frequencies because
of the term $k^2$, and this leads to an instability of the classical
gradient method (see \cite{mn08}). In \cite{mn08} the authors propose
a variant of the gradient method avoiding this problem. If
$a_k$ is the $k$-th Fourier coefficient at some iteration, we obtain
the new coefficient $a_k'$ for the successive step by adding
\begin{equation}
   \delta a_k = a_k'-a_k = -\delta\tau_k \frac{\partial A }{\partial a_k},
   \label{eq:iteration}
\end{equation}
and similarly for the coefficients $b_k$. This means that the decay
rate in the Fourier coefficients is controlled by a parameter
$\delta\tau_k$ which depends also on the order $k$ of the harmonic. If we set
\begin{equation}
   \delta\tau_k = \frac{T}{2(\pi k)^2} \delta,
   \label{eq:deltaTau}
\end{equation}
where $\delta > 0$ is a small positive constant, this removes the high
frequency instability.

To stop the iterations we could check the value of the residual
acceleration, i.e. the difference between the acceleration computed
from $u(t)$ and the force acting on the generating particle at time
$t$. However, this in practice can be done only when there are no
close encounters between the bodies.  In fact, when a passage near a
collision occurs, we have to choose a very large value of $F_M$ (see
\cite{simo:newfamilies}) to obtain a better approximation, and this
slows down the computations. For this reason we choose to stop the
iterations also when the increments $\delta a_k, \delta b_k$ become
small, as suggested in \cite{mn08}.

\subsection{Shooting method}
\label{s:shooting}
In order to refine the computation of the orbits, we use a shooting
method in the phase space of the generating particle, starting from the
Fourier polynomial approximations.  The goal is to solve the problem
\begin{equation}
   \begin{cases}
      \dot{x} = f(x), \\
      x(T/M) = Sx(0),
   \end{cases}
   \label{eq:bvp}
\end{equation}
where $x=(u,\dot{u}) \in \R^6$, $S$ is the matrix
\[
   S = 
   \begin{pmatrix}
      R & 0 \\
      0 & R
   \end{pmatrix},
\] 
with $R \in \mathcal{R}$ and $M > 0$ given by condition (c).  The
differential equation in \eqref{eq:bvp} comes from the Euler-Lagrange
equation of the functional defined by \eqref{action_formula}. Fixed
$n$ points $0 = \tau_0 < \tau_1 < \cdots < \tau_{n}=T/M$, we define
the function $G:\R^{6n} \to \R^{6n}$ as
\begin{equation}
   \begin{cases}
      G_i = \phi^{\tau_i - \tau_{i-1}}(x_{i-1}) - x_i, & i=1,\dots,n-1 \\
      G_n = \phi^{\tau_n-\tau_{n-1}}(x_{n-1}) - Sx_0,
   \end{cases}
   \label{eq:shootingEq}
\end{equation}
where the points $x_i = (u(\tau_i),\dot{u}(\tau_i))$ correspond to a
discretization of the position and velocity of the generating particle,
modeled with the Fourier polynomial approximation.
If we have a $T$-periodic solution $x(t)$ satisfying \eqref{eq:bvp},
the function $G$ evaluated at
\[
X = (x(\tau_0), \dots, x(\tau_{n-1}))
\]
vanishes. To search for the zeros of $G$ we use a least-squares
approach: we set
\[
F(X) = \frac{|G(X)|^2}{2},
\]
and search for the absolute minimum points by a modified Newton method.
The derivatives of $F$ are
\begin{gather}
   \frac{\partial F}{\partial x_j} = \sum_{i=1}^{n} \frac{\partial G_i}{\partial x_j}
   \cdot G_i,
   \label{eq:derGi} \\
   \frac{\partial^2 F}{\partial x_j \partial x_h} = \sum_{i=1}^n \bigg{[}
      \frac{\partial G_i}{\partial x_j}\frac{\partial G_i}{\partial x_h} +
      \frac{\partial^2 G_i}{\partial x_j \partial x_h}G_i
   \bigg{]}.
   \label{eq:der2Gi}
\end{gather}
The Jacobian matrix of $G$ is
\begin{equation}
   \begin{bmatrix}
      M_1 & -\Id & & \\
               & M_2 & \ddots & & \\ 
               & & \ddots & -\Id \\
               -S & & & M_n
   \end{bmatrix},
   \label{eq:jacG}
\end{equation}
where 
\[
   M_i = \frac{\partial}{\partial x} \phi^{\tau_{i}-\tau_{i-1}}(x_i).
\]
If $X'$ denotes the new value of $X$ at some iteration of the modified Newton method and
$\Delta X = X'-X$, at each step we solve the linear system
\begin{equation}
   A(X) \Delta X = - \frac{\partial F}{\partial X}(X),
   \label{eq:linSys}
\end{equation}
where the entries of the matrix $A$ are
\[
A_{jh} = \sum_{i=1}^n \frac{\partial G_i}{\partial
    x_j}\frac{\partial G_i}{\partial x_h},
\]
i.e. we consider only an approximation of the second derivatives \eqref{eq:der2Gi} of $F$.
However, $A$ is singular at the minimum points, since we are free to
choose the initial point along the periodic orbits. This degeneracy
can be avoided as in 
\cite{barrio:periodic}, by adding the condition on the first shooting point
\begin{equation}
   f(x_0) \cdot \Delta x_0 = 0,
   \label{eq:ortDisp}
\end{equation}
to \eqref{eq:linSys}, where $x_0, \Delta x_0$ are the first components
of $X, \Delta X$.  The system of equations \eqref{eq:linSys}, \eqref{eq:ortDisp} has
$6n+1$ equations and $6n$ unknowns, and we can solve it through the SVD
decomposition, thus obtaining the value of $\Delta X$.
We have performed the integration of the equation of motion and of the variational
equation with both the \texttt{DOP853} and the \texttt{RADAU IIA} integrators, available at
\cite{hairer:webpage}.

\section{Linear stability}
\label{s:ls}
We study the stability of the orbit of the generating particle, whose dynamics is defined
by \eqref{eq:bvp}. This corresponds to study the stability of the periodic orbit
of the full $N$-body problem with respect to symmetric perturbations. However, if the
orbit of the generating particle is unstable, also the full orbit of $N$-body is unstable.

From the standard Floquet theory we know that the monodromy matrix
$\mathfrak{M}(T)$ is a $6 \times 6$ real symplectic matrix with a double unit
eigenvalue, one arising from the periodicity of the orbit and the
other one from the energy conservation.
Since $\mathfrak{M}(T)$ is symplectic, we have only
three possibilities for the remaining eigenvalues
$\lambda_1,\lambda_2,\lambda_3,\lambda_4$:
\begin{enumerate}
   \item[1)] some of them are real and $\lambda_1\lambda_2 = 1, \lambda_3\lambda_4 = 1$;
   \item[2)] $\lambda_1,\lambda_2,\lambda_3,\lambda_4 \in \C \setminus \R$ and $\lambda_1 =
      \lambda_2^{-1} = \bar{\lambda}_3=\bar{\lambda}_4^{-1}$;
   \item[3)] $\lambda_1,\lambda_2,\lambda_3,\lambda_4 \in \C \setminus \R$ and $\lambda_1 =
      \lambda_2^{-1}=\bar{\lambda}_2, \lambda_3=\lambda_{4}^{-1}=\bar{\lambda}_4$. 
\end{enumerate}
As in \cite{kapela:eight}, we can give a stability criterion using the
values $T_1 = \lambda_1+\lambda_2$ and $T_2 = \lambda_3+\lambda_4$.
The characteristic polynomial of the monodromy matrix is
\[
\begin{split}
   p(\lambda) & =
   (\lambda-1)^2(\lambda-\lambda_1)(\lambda-\lambda_2)(\lambda-\lambda_3)(\lambda-\lambda_4)
   \\ & =
   (\lambda-1)^2(\lambda^2-T_1\lambda+1)(\lambda^2-T_2\lambda+1) \\ &
   = \lambda^6 - (T_1+T_2+2)\lambda^5 + (T_1T_2 + 2(T_1+T_2)
   +3)\lambda^4+\dots
\end{split}.
\]
Let us denote by $d_{ij}$ the generic entry of the monodromy matrix, and set
\[
a = \sum_{i=1}^6 d_{ii},\qquad b = \sum_{1 \leq i \leq j \leq
  6}(d_{ii}d_{jj}-d_{ij}d_{ji}).
\]
From the expressions of the coefficients of the characteristic
polynomial we obtain
\begin{equation}
   \begin{cases}
      T_1 + T_2 +2 = a, \\
      T_1 T_2 +2(T_1+T_2)+3 = b.
   \end{cases}
   \label{eq:relT1T2}
\end{equation}
It turns out that $T_1$ and $T_2$ are the roots of the polynomial of degree two
\begin{equation}
   q(s) = s^2 -(a-2)s + (b-2a+1).
   \label{eq:qPol}
\end{equation}
We use the following result, see~\cite{kapela:eight}.
\begin{lemma}
   Let $T_1,T_2$ the roots of the polynomial \eqref{eq:qPol}. The
   eigenvalues of the monodromy matrix lie on the unit circle if and
   only if
   \begin{equation}
      \begin{cases}
      \Delta = (a-2)^2-4(b-2a+1) > 0, \\
      \lvert T_1 \rvert < 2, \quad \lvert T_2 \rvert < 2.
   \end{cases}
      \label{eq:linStabCond}
   \end{equation}
   \label{lemma:rootsStab}
\end{lemma}
\begin{proof}
   The hypothesis $\Delta>0$ yields that $T_1, T_2$ are real and distinct.
   This excludes condition $2)$ above.
   We also exclude condition $1)$, because in this case we have
   \[
      \lvert T_1 \rvert = \lvert \lambda_1+\lambda_2 \rvert = \lvert
      \lambda_1+\lambda_1^{-1} \rvert > 2, \quad
      \lvert T_2 \rvert = \lvert \lambda_3+\lambda_4 \rvert = \lvert
      \lambda_3+\lambda_3^{-1} \rvert > 2.
   \]
   The only possibility is the third, that is the eigenvalues of the
   monodromy matrix lie on the unit circle.

\end{proof}

By this Lemma we avoid the numerical computation of the eigenvalues
and we can establish the stability of the orbit simply by computing
the roots of a polynomial of degree two, whose coefficients depend
only on the entries of the monodromy matrix $\mathfrak{M}(T)$.

Moreover, for symmetric periodic orbits we can factorize
$\mathfrak{M}(T)$ as in \cite{roberts:figure8}:
\begin{equation}
   \mathfrak{M}(T) = (S^T\mathfrak{M}({T}/{M}))^M,
   \label{eq:monFact}
\end{equation}
where $\mathfrak{M}(t)$ is the fundamental solution of the variational
equation at time $t$.  This means that we can integrate the
variational equation only over the time span $[0,T/M]$ and
we can study the stability by applying
Lemma \ref{lemma:rootsStab} to the matrix $S^T\mathfrak{M}(T/M)$.

Our numerical computations suggest that all the periodic orbits found
in Section \ref{s:enum} are unstable.  Non rigorous numerical values
of $\Delta, T_1, T_2$ for several cases can be found at
\cite{MFweb}. To make rigorous the results, we have to integrate the
equation of motion using interval arithmetic (\cite{moore:interval}),
as explained in the next section.

\section{Validation of the results}
\label{s:valid}
We use interval arithmetic to obtain rigorous estimates of the initial
condition of a periodic orbit and its monodromy matrix. Using these
estimates we can give a computer-assisted proof of the instability of
such orbit. To integrate rigorously a system of ODEs we use the
$C^1$-Lohner algorithm \cite{zgli:lohner} implemented in the CAPD
library \cite{CAPDweb}, which is based on a Taylor method to solve the
differential equations. Given a set of initial conditions and a final
time $\tau$, this algorithm produces an enclosure of the solution and
of its derivatives with respect to the initial conditions at time
$\tau$.

Denoting by $(u,\dot{u})\in \R^6$ the position and the velocity of the
generating particle, system \eqref{eq:bvp} has the first integral
of the energy
\[
   E(u, \dot{u}) = \frac{1}{2}|\dot{u}|^2 + \frac{1}{2}\sum_{R\in{\cal
       R}\setminus\{I\}}\frac{1}{|(R-I)u|}.
\]
Fix the value of the energy and use a surface of section to search for the
periodic orbit. More precisely, we use the Poincar\'e first return map
\[
   \mathsf{p} : \Sigma \to \Sigma,
\]
with
\[
   \Sigma = \{ (u,\dot{u}) \in \R^6 : u_3 = 0, \, E(u,\dot{u}) = h \},
\]
where $u_3$ is the third component of $u$ and $h$ is the value of the
energy of an approximated initial condition $(u_0,\dot{u}_0)$.  This
condition is computed from the solution $\tilde{u}(t)$ obtained with the shooting
method explained in Section~\ref{s:shooting}, which is propagated to
reach the plane $u_3=0$. Up to a rotation $R\in{\cal R}$, we can always assume
that $\tilde{u}(t)$ passes through this plane.

To compute an enclosure of the initial condition 
we use the interval Newton method (see, for instance,
\cite{moore:interval}). Given a box $B \subset \Sigma$ around
$(u_0, \dot{u}_0)$ we define the interval
Newton operator as
\[
   N( (u_0, \dot{u}_0), B, \mathsf{f}) = (u_0, \dot{u}_0) -
   [d\mathsf{f}(B)]^{-1}\mathsf{f}(u_0,\dot{u}_0),
\]
where $\mathsf{f}(u,\dot{u}) = \mathsf{p}(u,\dot{u})-(u,\dot{u})$ and $[d\mathsf{f}(B)]$ 
denotes the interval enclosure of $d\mathsf{f}(B)$. If we are able to verify that
\begin{equation}
N((u_0,\dot{u}_0),B,\mathsf{f})\subset B,
\label{inclusion}
\end{equation}
then from the interval Newton theorem there exists 
a unique fixed point in $B$ for the Poincar\'e map $\mathsf{p}$, and
therefore a unique initial condition for the corresponding periodic
orbit.

Given a box $B$ that satisfies \eqref{inclusion}, we use the
$C^1$-Lohner algorithm to compute an enclosure for
$\mathfrak{M}(T/M)$.
Then, using the factorization \eqref{eq:monFact}, we can compute also
an enclosure for the values of $\Delta, T_1,
T_2$ and verify whether the hypotheses of Lemma \ref{lemma:rootsStab}
hold or not.

\subsection{Numerical tests}
We have applied the method described above to give a rigorous proof of
the instability of the periodic orbits listed in Table
\ref{seqlist}.\footnote{For $\nu_{1}$ and $\nu_{43}$ we rotate the orbit of the
  generating particle, as it does not pass through the plane $u_3=0$.}
The $N$-body motion corresponding to the selected cases is displayed
in Figure~\ref{orbfig}.
 
\begin{table}[h!]
  \begin{center}
  \begin{tabular}{|c|c|c|}
     \hline
    label &$M$ &vertexes of $\mathcal{Q}_{\mathcal{O}}$\cr
    \hline
   $\nu_{1}$ &$3$ &{\small$[ 1, 3, 8, 10, 16, 5, 1 ]$} \\
    $\nu_{16}$ &$2$ &{\small$[ 1, 3, 8, 18, 13, 12, 4, 9, 2, 19, 11, 14, 1]$}  \\
    $\nu_{27}$ &$3$ &{\small$[ 1, 3, 7, 20, 18, 8, 15, 4, 6, 10, 16, 5, 1]$}  \\
   $\nu_{43}$ &$4$ &{\small$[ 1, 3, 8, 15, 4, 9, 2, 5, 1]$}\\
    \hline
  \end{tabular}
  \end{center}
  \caption{List of sequences used in the tests. The labels correspond
    to the enumeration used in the website \cite{MFweb}. }
  \label{seqlist}
\end{table}
\begin{figure}[!ht]
    \centerline{\includegraphics[width=5.5cm]{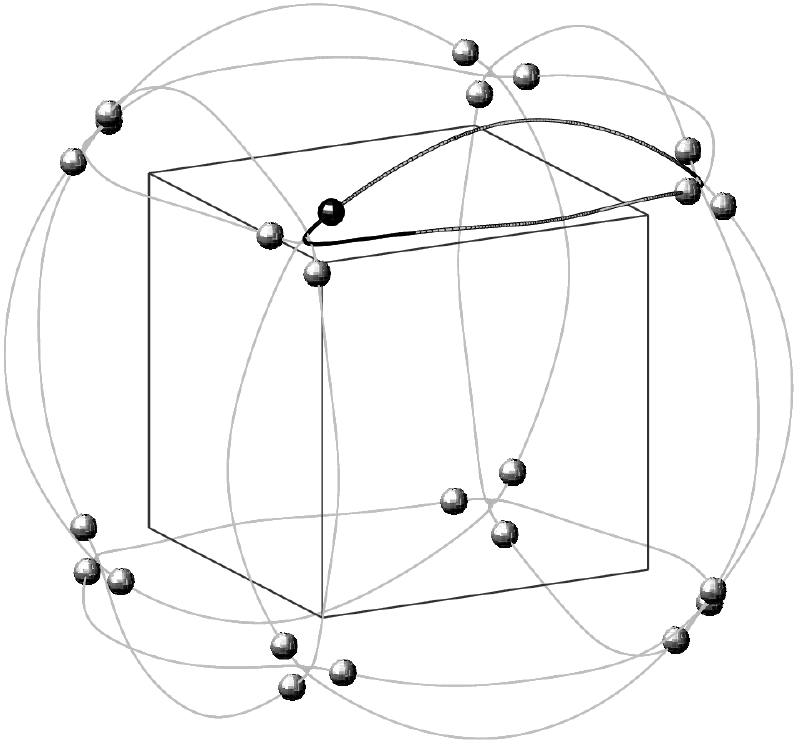}\hskip 1cm
    \includegraphics[width=5cm]{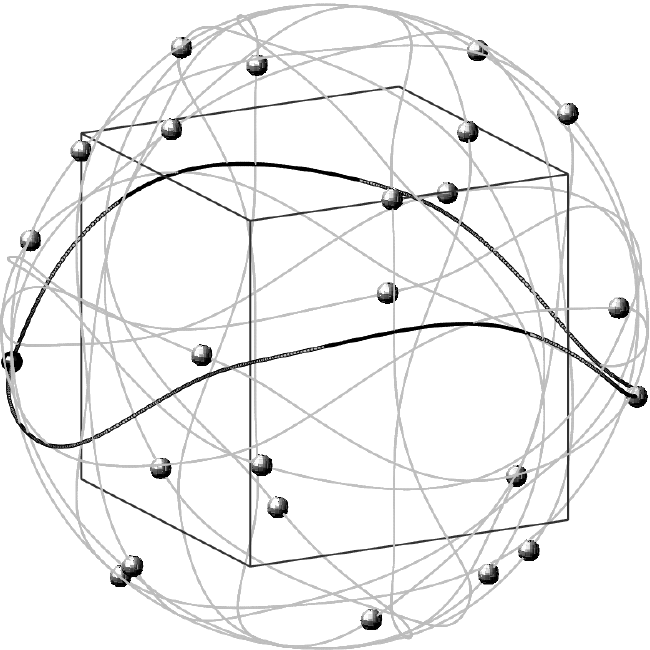}}
  \vskip 0.4cm
    \centerline{\includegraphics[width=5cm]{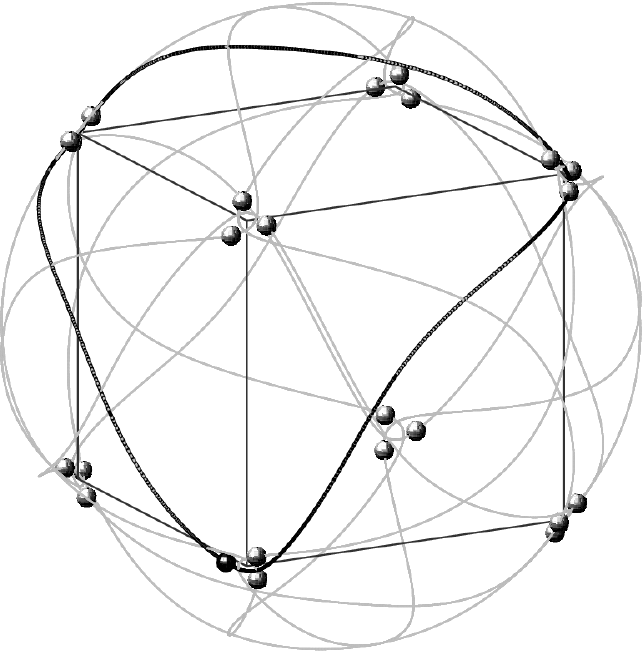}\hskip 1cm
    \includegraphics[width=5.5cm]{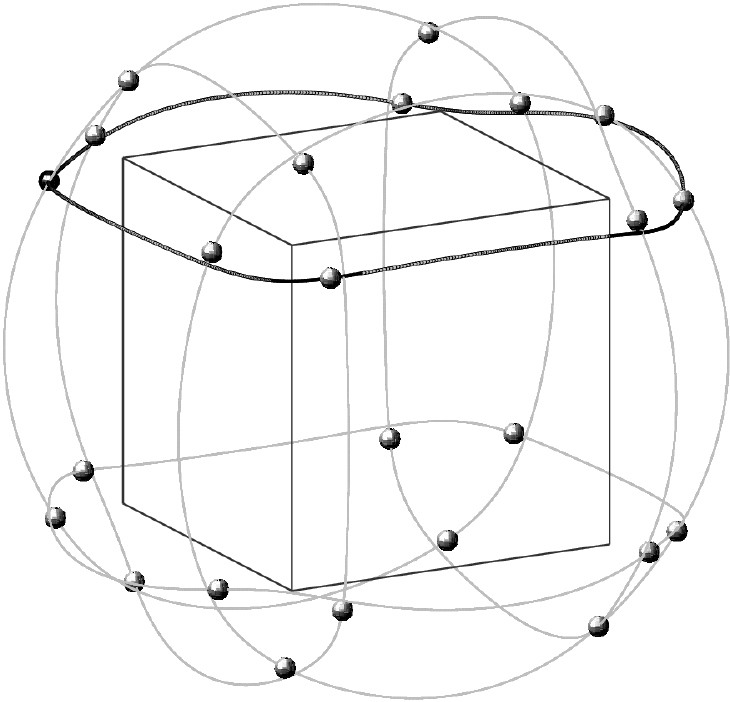}}
  \caption{Periodic motions of the $N$ bodies corresponding to the
  sequences listed in Table~\ref{seqlist}. The solid black curve represents the trajectory
  of the generating particle.}
\label{orbfig}
\end{figure}

Hereafter we shall use a notation similar to \cite{kapela:eight, kapela:KAM} to describe an interval: first we write the digits
shared by the interval extrema, then the remaining digits are reported
as subscript and superscript. Thus, for instance, we write
\[
  12.3456789_{1234}^{5678} 
\]
for the interval
\[
  [12.34567891234, 12.34567895678].
\]
For the four selected cases in Table \ref{seqlist} we checked that
condition \eqref{inclusion} holds using multiple precision interval
arithmetic. This ensures the existence of an initial condition for the orbits
in the selected box $B$. When the bodies
do not undergo close approaches, as in the cases of $\nu_{16}$ and
$\nu_{43}$, the inclusion can be checked with a much larger box.  On
the other hand, when close approaches occur, as in the cases of
$\nu_{1}$ and $\nu_{27}$, we are forced to use a very tiny box, a
longer mantissa
and a higher order for the Taylor method (see
Table~\ref{tabresIA}). This increases significantly the computational
time.
To show the instability of these orbits we prove that conditions
\eqref{eq:linStabCond} are violated.  For $\nu_{16}$ and $\nu_{43}$
these computations are successful even using only double precision
interval arithmetic. In the other two cases we need multiple precision
and a higher order for the Taylor method. As we can see from Table
\ref{tabresIA}, the value of $T_1$ is well above $2$, that yields a
computer-assisted proof of the instability of these four test orbits.
In Table~\ref{tabres} we report also the non-rigorous values of
$\Delta, T_1, T_2$, obtained by numerical integration without interval
arithmetic. Comparing the values in the two tables, we can see that
the non-rigorous ones are in good agreement with the estimates
computed with interval arithmetic. Non-rigorous values for several
other orbits can be found at \cite{MFweb}.

\begin{table}[!ht]
  \begin{center}
  \begin{tabular}{|c|c|c|c|c|c|c|}
     \hline
    label  & mantissa (bits) &size($B$) &order & $\Delta$ &$T_1$ &$T_2$\cr
    \hline
    $\stackrel{\phantom{R}}{\nu_{1}}$  & $100$ & $2\cdot 10^{-18}$ & 30  &$1488.95_{2965}^{3031}$ &$43.365_{499}^{500}$ &$4.778546_1^6$ \\
    \hline
    $\stackrel{\phantom{R}}{\nu_{16}}$  & $52$ & $2\cdot 10^{-14}$ &15 & $90582.1_{06}^{30}$ &$301.0993_{68}^{89}$ &$0.1307_{359}^{566}$ \\
                        & $\stackrel{\phantom{}}{100}$ & $2\cdot 10^{-14}$ &30 & $90582.1_{14}^{22}$ &$301.0993_{76}^{82}$ &$0.13074_{33}^{92}$ \\
\hline
    $\stackrel{\phantom{R}}{\nu_{27}}$  & $100$ & $2\cdot 10^{-25}$ &30 & $5105.47178_{6}^{7}$ &  $73.279035_{5}^{6}$  &  $1.826451_{3}^{4}$ \\
\hline
    $\stackrel{\phantom{R}}{\nu_{43}}$  & $52$ & $2\cdot 10^{-14}$ &15 & $7.0355_{19}^{20}$ & $9.2322605_1^9$  &  $6.579805_0^1$ \\
                        & $\stackrel{\phantom{}}{100}$& $2\cdot 10^{-14}$ &30 & $7.0355_{19}^{20}$ &$9.2322605_{3}^{8}$ &$6.579805_{0}^{1}$ \\
    \hline
  \end{tabular}
  \end{center}
  \caption{Enclosures for the values of $\Delta, T_1, T_2$.}
  \label{tabresIA}
\end{table}

\begin{table}[!ht]
  \begin{center}
  \begin{tabular}{|c|c|c|c|}
     \hline
    label  & $\Delta$ &$T_1$ &$T_2$\cr
        \hline
        {\small$\nu_{1}$}  & $1488.953003$ &  $43.365500$ &  $4.778546$ \\
    \hline
        {\small$\nu_{16}$}  & $90582.118054$ &  $301.099379$ &  $0.130746$ \\
        \hline
            {\small$\nu_{27}$}  & $5105.471786$  &  $73.279035$  &  $1.826451$ \\
        \hline
    {\small$\nu_{43}$}  & $7.035519$     &  $9.232260$   &  $6.579805$ \\
    \hline
  \end{tabular}
  \end{center}
  \caption{Non rigorous values of $\Delta, T_1, T_2$.}
  \label{tabres}
\end{table}

\section{Conclusions and future work}

Using the algorithm described in Section \ref{s:enum}, we created a
list of all the periodic orbits of the $N$-body problem whose
existence can be proved as in \cite{fgn10}, where only a few of them
were listed. We also set up a procedure aimed to compute these orbits, 
and investigated the stability for a large number of them.
All the solutions found with the rotation groups $\{\mathcal{T},
\mathcal{O}\}$ appear to be unstable, with a large value of $|T_1|$ or
$|T_2|$ (see the website \cite{MFweb} for the results).
Using multiple precision interval arithmetic, we were able to make
rigorous these results for a few orbits, producing a computer-assisted
proof of their instability.
From the numerical point of view, the main difficulty is to automatize
the choice of the parameters appearing in the computations: the order
of the Fourier polynomials, the number of shooting points, the size of
the boxes, the order of the Taylor method, the mantissa size for 
multiple precision computations, etc.
Moreover, when the bodies undergo close approaches a longer
computational time is needed to check whether condition
\eqref{inclusion} holds, because a larger size of the mantissa is
required.
This requires a long computational time. For these reasons we
performed interval arithmetic computations only for a few orbits in
our list.

We have also to point out that there is no guarantee that the computed
periodic solutions
correspond to minimizers of the action $\A$, whose existence have been
assessed in Sections~\ref{s:exist}, \ref{s:enum}. Indeed, with the
proposed procedure, we first search for a minimizer of $\A$ in a
finite dimensional set of Fourier polynomials by a gradient method,
decreasing the value of the action; then we refine these solutions by
a shooting method. This algorithm is meant to obtain local minima of
the action $\A$.  However, there is no proof, not even
computer-assisted, that the computed solutions minimize the action
$\A$ in the cone $\K$.  We plan to investigate this aspect in the future.

As a final remark, we observe that the procedure described in
Section \ref{s:valid} can also be used to prove the existence of
periodic orbits not included in our list.  For example, the sequence
\[
   \nu = [ 1, 3, 7, 18, 20, 24, 12, 4, 9, 17, 19,
     21, 23, 14, 1 ],
\]
of vertexes of $\mathcal{Q}_{\mathcal{O}}$ does not satisfy the
condition on the maximal length \eqref{eq:actionPiecewise}: in fact it
has $M=2$, $l_{\text{max}}(2) = 12$ and the length of $\nu$ is $14$.
This means that we cannot exclude total collisions.  However, the
validity of condition \eqref{inclusion} can be checked numerically.
Using the approximated orbit obtained with multiple shooting, we were
able to check that \eqref{inclusion} holds using a box with size $2
\cdot 10^{-14}$. In this way, we have a computer-assisted proof of the
existence of a periodic orbit belonging to $\K(\nu)$, represented in
Figure~\ref{fig:concl}.

\begin{figure}[h!]
  \centerline{\includegraphics[width=5cm]{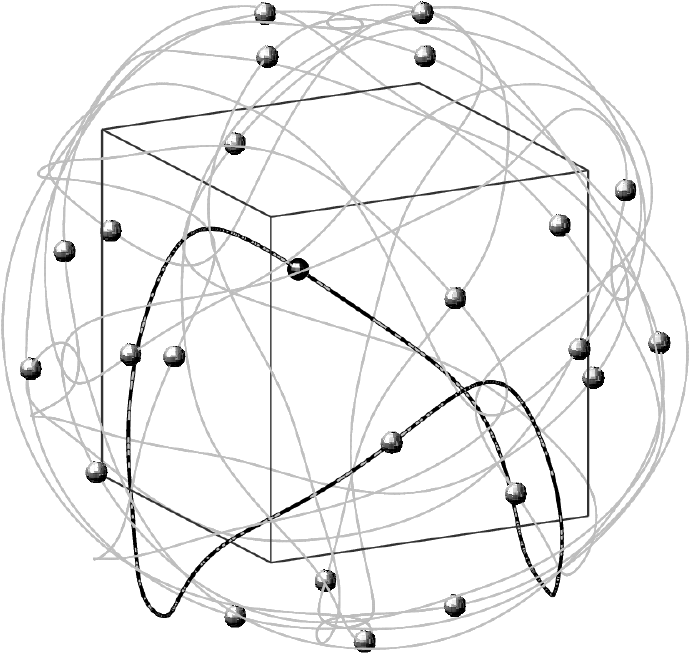}}
  \caption{Periodic motion with $\nu$ not fulfilling condition
    \eqref{eq:actionPiecewise}.}
  \label{fig:concl}
\end{figure}

We also found the values
\[
\Delta=4423_{29}^{58}, \quad T_1 =665.3_{83}^{94}, \quad T_2=0._{293}^{305}, 
\]
that yields a rigorous proof of the instability of this orbit.

\section*{Acknowledgements}
We wish to thank T. Kapela and O. van Koert for their useful suggestions
concerning this work.
Both authors have been partially supported by the University of Pisa
via grant PRA-2017 `Sistemi dinamici in analisi, geometria, logica e
meccanica celeste', and by the GNFM-INdAM (Gruppo Nazionale per la
Fisica Matematica).

\bibliography{platorb_ls}{} 

\end{document}